\documentclass[sigconf]{acmart}

\usepackage{balance}

\def\BibTeX{{\rm B\kern-.05em{\sc i\kern-.025em b}\kern-.08emT\kern-.1667em\lower.7ex\hbox{E}\kern-.125emX}}
    
\newcommand\ie{i.\,e.}
\newcommand\eg{e.\,g.}

\usepackage{enumitem}

\usepackage{amsthm}
\usepackage[titletoc,title]{appendix}
\newtheorem{theorem}{Theorem}[section]
\newtheorem{lemma}[theorem]{Lemma}
\newcommand\define{\ensuremath{\mathrel{\stackrel{\mathrm{def}}{=}}}}

\DeclareMathOperator*{\argmin}{arg\,min}

\renewcommand\footnotetextcopyrightpermission[1]{} 
\begin{document}
\fancyhead{}

\title{Personalized Purchase Prediction of Market Baskets with Wasserstein-Based Sequence Matching}

\author{Mathias Kraus}

\affiliation{%
	\institution{ETH Zurich}
	\streetaddress{Weinbergstr. 56/58}
	\city{8092 Zurich}
	\state{Switzerland}
}
\email{mathiaskraus@ethz.ch}

\author{Stefan Feuerriegel}
\affiliation{%
	\institution{ETH Zurich}
	\streetaddress{Weinbergstr. 56/58}
	\city{8092 Zurich}
	\state{Switzerland}
}
\email{sfeuerriegel@ethz.ch}

\renewcommand{\shortauthors}{Kraus and Feuerriegel}

\begin{abstract}
Personalization in marketing aims at improving the shopping experience of customers by tailoring services to individuals. In order to achieve this, businesses must be able to make personalized predictions regarding the next purchase. That is, one must forecast the \emph{exact} list of items that will comprise the next purchase, \ie, the so-called market basket. Despite its relevance to firm operations, this problem has received surprisingly little attention in prior research, largely due to its inherent complexity. In fact, state-of-the-art approaches are limited to intuitive decision rules for pattern extraction. However, the simplicity of the pre-coded rules impedes performance, since decision rules operate in an autoregressive fashion: the rules can only make inferences from past purchases of a \emph{single} customer without taking into account the knowledge transfer that takes place between customers. 

In contrast, our research overcomes the limitations of pre-set rules by contributing a novel predictor of market baskets from sequential purchase histories: our predictions are based on similarity matching in order to identify similar purchase habits among the complete shopping histories of \emph{all} customers. Our contributions are as follows: (1)~We propose similarity matching based on subsequential dynamic time warping (SDTW) as a novel predictor of market baskets. Thereby, we can effectively identify cross-customer patterns. (2)~We leverage the Wasserstein distance for measuring the similarity among embedded purchase histories. (3)~We develop a fast approximation algorithm for computing a lower bound of the Wasserstein distance in our setting. An extensive series of computational experiments demonstrates the effectiveness of our approach. The accuracy of identifying the exact market baskets based on state-of-the-art decision rules from the literature is outperformed by a factor of 4.0. 
\end{abstract}

\keywords{purchase prediction, market basket, dynamic time warping, sequence matching, Wasserstein distance, product embeddings}

\maketitle

\section{Introduction}

Understanding and predicting consumer decision-making has been subject to extensive research. The outcome of the decision-making process forms a multi-category shopping basket which comprises the complete set of items that an individual consumer has purchased together, \ie, the so-called \emph{market basket} \cite{Guidotti.2017}. Both online and offline merchants were traditionally interested in understanding the composition of customers' market baskets, since it enables them to gain valuable insights that can inform personalized marketing and targeted cross-selling programs. These efforts were strengthened in the recent wave of personalization in marketing and has fostered a variety of predictive applications. As a result, firms aim at predicting personalized market baskets from the next purchase of individuals, based on which they improve customer service, supply chain management, or assortment optimization \cite{Agrawal.2017}. Despite the prosperous outlook for firm operations, actual works on forecasting market baskets are scarce.  

The problem of predicting market baskets entails clear differences from other prediction tasks in marketing. To this end, forecasting techniques have been applied to \emph{sales} \cite{Choi.2011,Choi.2014,Winters.1960}, yet where, different from market basket predictions, purchases are aggregated across stores. Specifically, sales forecasting involves a multivariate time series as input, where the output is then a single value denoting the overall sales volume. Hence, this prediction task is considerably simpler than inferring the dynamic nature of subsets with variable size. Because of this difference, sales forecasting has been approached by feature-based classifiers and recurrent neural networks \cite{Kraus.2018,Loureiro.2018}; however, these methods are not applicable to the outcome variable in our research. When making predictions with regard to assortments, recent works have been concerned with \emph{item-level} predictions \cite{Jacobs.2016}, where the purchases for a single product are predicted and not a complete basket. \emph{User-level} predictions of purchases \cite{Lichman.2018} are limited to a single time series as input, as well as a univariate output, whereas market baskets must be modeled as a dynamic set. \emph{Ranking}, such as in (session-based) recommender systems, orders candidate items by the probability of purchase \cite{Hidasi.2015, Hidasi.2018}; however, this task is merely supposed to re-arrange a list of items \cite{Rendle.2009, Weston.2013}, thus failing to provide a subset. Altogether, it becomes evident that the specific nature of predicting market baskets is subject to unique characteristics.

\textbf{Problem statement:} The problem of predicting market baskets refers to identifying the exact set of items a customer will purchase in her next transaction \cite{Guidotti.2018}. Formally, we consider a retailer with $n$ available items for purchase, which are given by $\mathcal{I} = \{ i_1, \ldots ,i_n\}$. Furthermore, the retailer has already served a set of $k$ customers $C = \{c_1, \ldots, c_k\}$ for who purchase histories have been recorded. The objective is then to infer the next purchase for a given customer $c$. Additional information for this customer $c$ is available, namely, the ordered sequence of previously purchased market baskets, \ie, $B_c = [b_c^{1}, b_c^{2}, \ldots, b_c^{m_c}]$, where $b_c^{i} \subseteq \mathcal{I}$, $i = 1, \ldots, m_c$, represents the basket composition and $m_c$ is the number of past transactions. Analogously, the complete set of purchase histories across the whole customer base $1, \ldots, k$ is given by $\mathcal{B} = \{B_{1}, B_{2}, \ldots, B_{k} \}$. Note that the size of the market basket $\left|b_{c}^{i} \right|$ and the length of the purchase history are variable, and depend on both $i$ and $m_c$. Given the purchase history $B_c$ of customer $c$ with $m_c$ transactions, the market basket prediction returns the expected basket $b_c^\ast \subseteq I$. This should provide the \emph{exact} set $b_c^{m_c+1}$ of items that customer $c$ will purchase with the next transaction $m_c+1$. In this sense, not a single item is returned but it is supposed to be a subset of $I$. Therefore, only \emph{identical} matches between items of $b_c^\ast$ and $b_c^{m_c+1}$ count as correct predictions. 

The stated problem is naturally challenging: market baskets are highly variable and subject to few recurrent patterns due to repeated purchases or co-purchases. At the same time, the number of available items for purchase, $n$, will be substantially larger than the average size of a market basket, \ie, $\left|b_c^{i}\right| \ll n$, thus turning into an extreme subset selection problem. This problem thus forbids a straightforward application of recurrent neural networks or ranking approaches which is why prior research has struggled with this type of prediction. In general, predictions usually involve only a single-valued output, whereas, in market basket predictions, the output is given by a subset. Needless to say, predictions of subsets are generally rare. 

The common approach to predicting market baskets draws upon association rules. Generally, these encode certain repeated purchases or co-purchases into simple decision rules, such as, \eg, \mbox{IF \{ pizza, beer \}} \mbox{THEN \{ painkiller \}}, where the then-clause states the expected subsequent purchase. The generation of association rules is commonly based on the Apriori algorithm \cite{Agrawal.1994} or variants thereof. In keeping with this, recent methods build upon extensive domain knowledge, together with manual feature rule engineering, in order to devise tailored means for finding higher-order purchase patterns \cite{Guidotti.2017}. Although intuitive, we argue that there is considerable capacity for improvement, so that a higher prediction performance can be attained: 
\begin{enumerate}[leftmargin=0.5cm]
	\item The proposed prediction rules are all based on exact product matching; \ie, for decision rules, the products ``white wine'' and ``red wine'' are completely different as rules cannot learn an implicit hierarchical structure from the data. However, presumably, the inherent similarity between some products (\ie, substitutes) is also responsible for a similar influence on subsequent purchases. 
	\item Rules are autoregressive in the sense that they can only make inferences from past purchases of a \emph{single} customer without taking into account the knowledge transfer that takes place between customers.
	\item Rule-based approaches fail to provide a measurement for the predictability of the next market basket. However, this is often required by merchants, as such decision support should target only customers with sufficient confidence in the prediction. 
\end{enumerate}

\textbf{Proposed prediction algorithm:}\footnote{Code available from \url{https://github.com/mathiaskraus/MarketBasket}}. In order to address the aforementioned limitations, we propose a novel prediction algorithm that is tailored to market baskets. Specifically, our algorithm allows to identify co-occurrences between shopping histories from different customers. Furthermore, each product is represented by a multi-dimensional (embedded) vector in order to learn similarity structures among the assortment. For this reason, we develop an algorithm that performs similarity matching across all shopping histories, $\mathcal{B}$. Formally, our proposed algorithm is given by a combination of $k$-nearest neighbor and subsequence dynamic time warping~($k$NN-SDTW) that operates on embedded product vectors. This form of sequence-based similarity matching is computed according to the Wasserstein distance. Thereby, we interpret market baskets as probability distributions of products from an assortment. We further develop an efficient lower bound of the Wasserstein distance with theoretical guarantees. These components together allow us to make accurate inferences with regard to future market baskets.  

\textbf{Theoretical properties:} Our algorithmic advances are also manifested in both the computational complexity and lower bounds of the Wasserstein distance. On the one hand, we perform subsequence matching based on the dynamic time warping algorithm, \ie, subsequential dynamic time warping (SDTW), in order to accelerate the computation of distances between shopping histories. In fact, the computational cost of our algorithm SDTW is $O(n \, (m+1))$. This approximation is superior to a straightforward application of dynamic time warping, which attains merely $O(n^2 \, m)$. On the other hand, we propose a tight lower bound for the $p$-Wasserstein distance ($p \geq 1$). The resulting algorithm requires $O(l^2)$ computations as compared to $O(l^3 \, \log{l})$ for the na{\"i}ve computation, where $l$ denotes the number of unique items in the market baskets under consideration.

\textbf{Evaluation:} The effectiveness of our prediction algorithm is evaluated based on three real-world datasets, which vary in size between 9,000 and over 60,000 customers. We also adapt to the specific needs of offline and online retailers by evaluating our prediction algorithms across multiple levels of granularity, namely, aisle-level and product-level. Overall, we find consistent improvements of our learning algorithm over association rules and na{\"i}ve heuristics from practice. While our approach outperforms baseline models by 2.54\,\% on a multi-category dataset comprising of groceries, it multiplies the ratio of correct predictions by a factor of 4.0 for a dataset covering food, office supplies and furniture. 

\textbf{Contributions:} There are only a few studies on predicting market baskets and these are limited to the theoretical shortcomings of association rule learning; hence, we advance the literature in the following directions:
\begin{enumerate}[leftmargin=0.5cm]
	\item We aim to overcome the current limitations in market basket prediction: our novel prediction algorithm can learn hidden structures among products; it can draw upon the complete database of trajectories during prediction time and thus leverages cross-customer knowledge. In contrast to the state-of-the-art, our numerical experiments establish that our algorithm can successfully learn the underlying structure. This is best seen in the fact that we have a several magnitudes higher rate of exact predictions, implying that simple rules are ineffective for encoding the structure behind sequential purchases of market baskets.
	\item This work provides the first combination of dynamic time warping for subsequence matching and the Wasserstein distance. The latter provides a natural similarity score that can relate subsets of different sizes to one another. 
	\item We achieve scalability by deriving a fast variant of subsequence matching based on dynamic time warping with an explicit lower bound, thus giving rise to a tight $p$-Wasserstein bound. 
\end{enumerate}

\section{The Proposed Prediction Algorithm for Market Baskets}

We introduce the following mathematical notation. Given a sequence $B = [ b_1, \ldots, b_m ]$, we then refer to a subsequence of $B$ that spans from the $i_s$-th to the $i_e$-th item via $B[i_s: i_e] \define [b_{i_s}, \ldots, b_{i_e}]$. Analogously, we use $B[i]$ as a short-hand form for $b_i$. 

\subsection{Overview}

Our approach performs along four steps (see Fig.~\ref{fig:framework}): 
\begin{enumerate}[leftmargin=0.5cm]
	\item \textbf{Product embeddings:} First, we build item embeddings, where each item is represented by a multi-dimensional vector with similar items being closer together (in terms of cosine similarity). For instance, ``white wine'' and ``red wine'' are mapped on a representation that is in closer proximity than ``white wine'' and ``apples''. 
	\item \textbf{Wasserstein-based purchasing similarity:} Next, we utilize the Wasserstein distance in order to compute distances between market baskets. This allows us to compute distances over sets with a variable number of items. Later, it reveals another benefit, as it allows us to interpret a market basket as a probability distribution over products from a given assortment. In this sense, the distance between two market baskets is then defined by the minimum cost of turning one probability distribution of products into the other. As we shall see later, we suggest an efficient approximation scheme via a lower bound in order to accelerate the computation and, only if the lower bound is promising, perform an exact computation. 
	\item \textbf{Most similar purchasing histories:}  Based on the Wasserstein distance, we build a $k$-nearest neighbor sequence matching by employing a subsequence dynamic time warping~(\ie, $k$NN-SDTW). This approach computes distances between sequences of market baskets. This further gives a numerical score that describes the predictability of a customer's market basket.\\  
	Our model builds upon a $k$-nearest neighbor algorithm which locates the most similar purchase histories as follows. Given is the purchase history $B_c$ of a customer $c$ for whom we want to predict the next market basket. Mathematically, we then try to find a similar (potentially subsequential) purchase history $B_d[i_s: i_e]$ across all customers $d \in \{ 1, \ldots, k \}$, where $i_s$ and $i_e$ denote the customer-specific start and end of the subsequence, respectively. Given the most similar purchase history, we now have $B_{d}[i_e + 1]$ as a set of potential items that customer $c$ will purchase in the next transaction. 
	\item \textbf{Market basket prediction:} Finally, we can make a prediction of the next market basket. Here we simply choose the next market basket from the most similar purchase histories; however, with a small adaptation. We again compute the Wasserstein distance and, if that exceeds a threshold $\tau$, revert to a fallback prediction. In other words, we leverage the fact that our algorithm knows in which situations the sequence matching is uncertain. If the threshold $\tau$ is exceeded, we predict the top-$n_c$ items from the purchase history $B_c$, where $n_c$ is the number of the items in the previous purchase, \ie, $|b_c^{m_c}|$.
\end{enumerate}
Our approach advances the body of research in the following directions. First, we utilize product vectors to model similarity between products. Second, we measure the distance between market baskets (\ie, sets of product embeddings) based on the Wasserstein distance. Finally, we apply an extension of dynamic time warping in order to locate the most similar customer (\ie, the nearest neighbors). The following sections describe each step in detail. 

\begin{figure}	
	\centering
	\includegraphics[width=.8\linewidth]{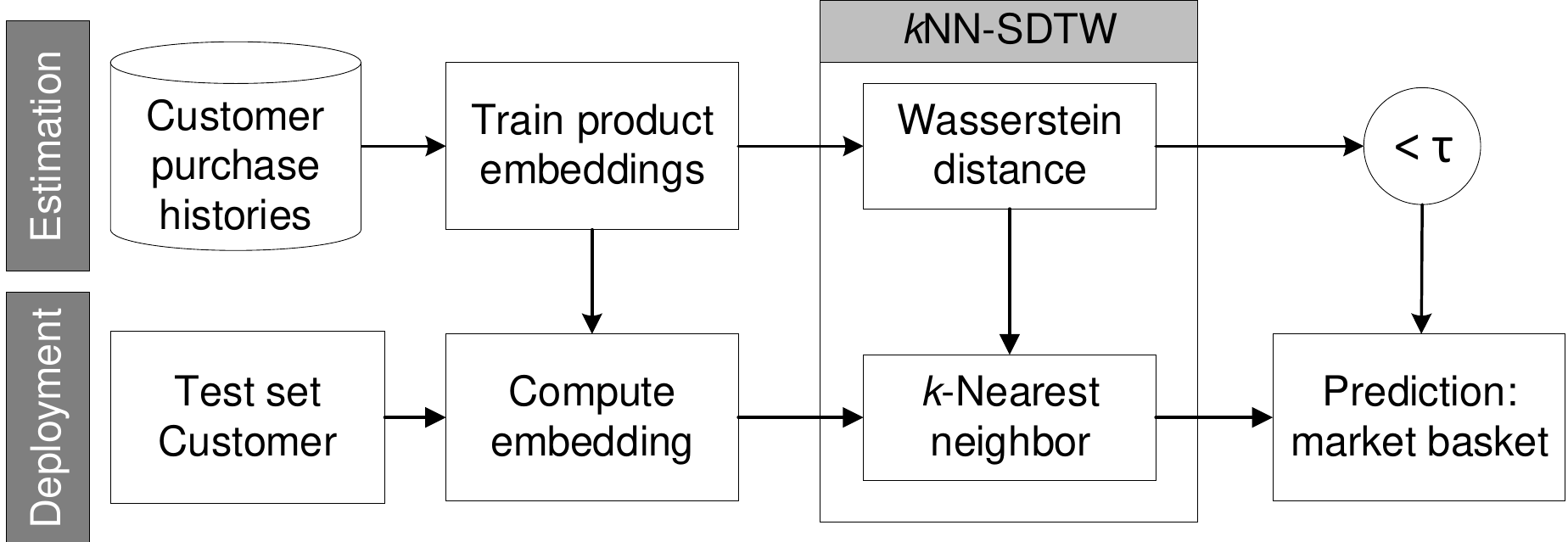}
	\caption{Schematic process for predicting the next market basket. At deployment, we first locate the most similar \mbox{(sub-)sequences} among all purchase histories in the customer base. This yields candidate products based on which we predict next purchase. Depending on the estimated similarity from the Wasserstein distance being below a threshold $\tau$, the candidate vector is used or the process reverts to a simple heuristic instead.}
	\label{fig:framework}
\end{figure}

\subsection{Item Embeddings}

Our model builds upon item embeddings, \ie, multi-dimensional vectors that represent products in a vector space in which similar products are closer to each other (according to cosine similarity). Mathematically, the embeddings translate the high-dimensional sparse vectors that denote products in one-hot fashion into a lower-dimensional dense vector. In order to compute these item embeddings, we adapt neural embeddings which were originally introduced in natural language processing \cite{Mikolov.2013}. The method aims at finding item representations that capture the relation between an item to other items that are present in the same market basket. 

Formally, we maximize the log probability as follows. Given a market basket $b_c^i \subseteq I$, we aim at maximizing
\begin{align}
\sum_{p \in b_c^i} \sum_{\substack{q \in b_c^i \\ q \neq p}} \log \Pr (p \, | \, q),
\end{align}
where $p,q$ refer to products in the market basket $b_c^i \subseteq I$. Here $\Pr (p \, | \, q)$ denotes the softmax function
\begin{align}
\Pr (p \, | \, q) = \frac{\exp(u_p^T v_q)}{\sum_{r \in I} \exp(u_p^T v_r)}
\end{align}
with $u_p \in \mathbb{R}^m$ and $v_q \in \mathbb{R}^m$ being latent vectors corresponding to the target and context representation of product $p$. By averaging $u_p$ and $v_q$, we obtain an $m$-dimensional vector - the embedding vector of product $p$. The choice of the hyperparameter $m$ is detailed in our online appendix.

\subsection{Wasserstein Distance}

Embeddings link to the distance between products, while we now define a distance between market baskets. Note that two market baskets comprise a different number of products (\ie, $m$ and $n$) and their sizes must not necessarily be the same (\ie, $m \neq n$). Hence, we require a distance over variable-size sets.  

\subsubsection{Definition}

Following \cite{Tran.2016}, we adapt the Wasserstein distance to our setting. Let $X = \{x_1, \ldots, x_m\}$ and $Y = \{y_1, \ldots, y_n\}$ refer to two market baskets which are given by a different set of product embeddings. Then the Wasserstein distance of order $p \geq 1$ is defined as
\begin{align}
d_{\mathrm{W}}^{(p)} (X, Y) \define \min_C \left( \sum_{i=1}^{m} \sum_{j=1}^{n} c_{ij} \, d(x_i, y_j)^p \right)^{\frac{1}{p}},
\end{align}
where $d$ is a metric and $C$ is an $m \times n$ transportation matrix with elements $c_{ij}$ that satisfy
\begin{align}
c_{ij} &\geq 0 &&\text{for all } 1 \leq i \leq m \text{, and } 1 \leq j \leq n, \\
\sum_{j=1}^{n} c_{ij} &= \frac{1}{m} &&\text{for all } 1 \leq i \leq m, \label{constraint1} \\
\sum_{i=1}^{m} c_{ij} &= \frac{1}{n} &&\text{for all } 1 \leq i \leq n. \label{constraint2}
\end{align}
Note that $X = \emptyset$ (or $Y = \emptyset$) implies $d_{\mathrm{W}}^{(p)} (X, Y) = \infty$. However, this corresponds to a setting with an empty market basket, which does not appear in our work.

Fig.~\ref{fig:item_embeddings} illustrates the idea underlying the choice of the Wasserstein distance. It is favorable in our scenario of measuring similarity between market baskets, as it can measure similarities between a varying number of products. For instance, Fig.~\ref{fig:item_embeddings} illustrates a market basket B which contains ``Sprite'' and ``Coca Cola Cherry''. The latter is related to ``Coca Cola'' from market basket A, yet it is not included in this market basket. Therefore, both market baskets A and B share not overlapping items and, on top of that, are of different size. However, their Wasserstein distance over product embeddings suggests that they are fairly similar. 

\begin{figure}
	\centering
	\includegraphics[width=.8\linewidth]{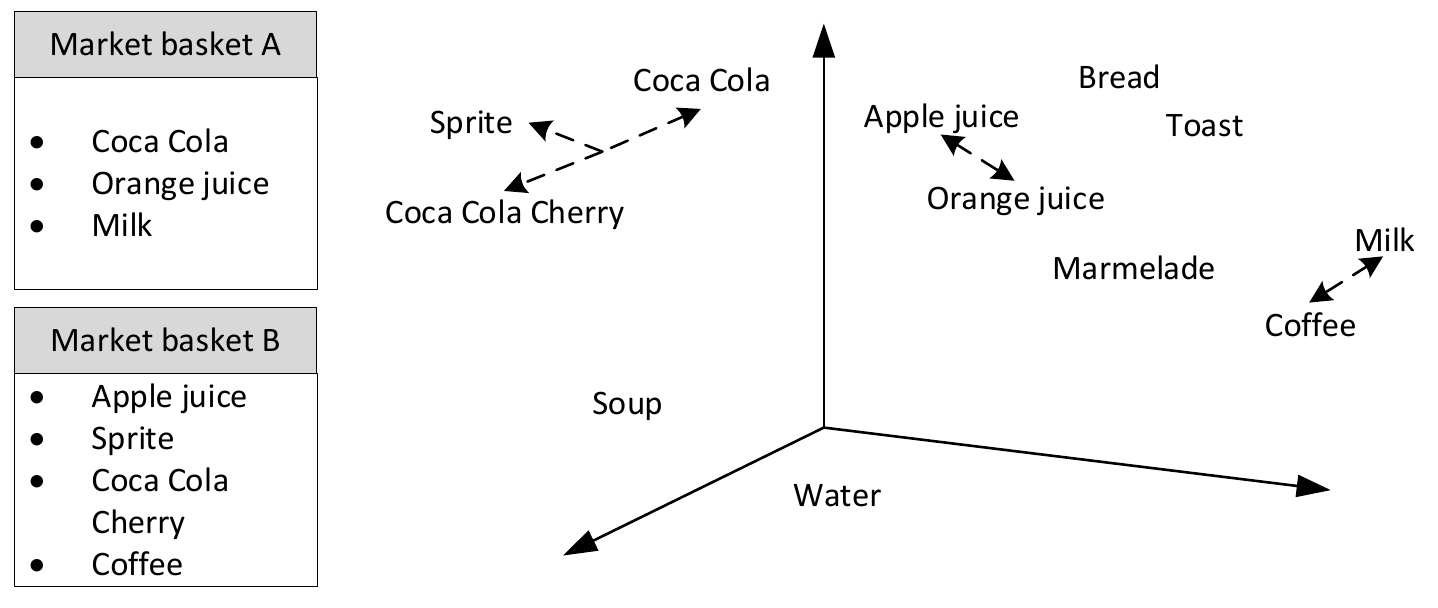}
	\caption{Two market baskets A and B with no overlapping items and different sizes. Common distance measures would return a large distance between the two baskets. In contrast, the Wasserstein distance between A and B suggests that both are fairly similar, as it can learn products that act as substitutes (as indicated by dashed arrows).}
	\label{fig:item_embeddings}
\end{figure}

The best average time complexity of solving the Wasserstein distance is $O(l^3 \log(l))$, where $l$ denotes the number of unique items in the union $X \cup Y$~\cite{Pele.2009}. This time complexity renders applications infeasible that require large-scale comparisons such as ours. In fact, we later need to compute distances among thousands of market baskets and, hence, we suggest an efficient computation scheme in the following.

\subsubsection{Lower Bound of the Wasserstein Distance}

Instead of computing the exact Wasserstein distance for each sample, we first compute a lower bound. This allows us to eliminate candidate sequences when searching for similar purchase histories and, as we shall see later, is responsible for considerably reducing the runtime. Following \cite{Kusner.2015}, we develop a lower bound to $d_{\mathrm{W}}^{(p)} (X, Y)$ but generalize their finding to arbitrary Wasserstein distances in the following lemma. 

\begin{lemma}
	Given $X = \{x_1, \ldots, x_m\}$ and $Y = \{y_1, \ldots, y_n\}$, it holds that 
	\begin{align}
	d_{\mathrm{W}}^{(p)} (X, Y) \geq \sum_{i=1}^{m} \min_{k=1, \ldots, n} d(x_i, y_k)^p \frac{1}{m} \label{lower_bound1}
	\end{align} 
	and 
	\begin{align}
	d_{\mathrm{W}}^{(p)} (X, Y) \geq \sum_{j=1}^{n} \min_{k=1, \ldots, m} d(x_k, y_j)^p \frac{1}{n}, 
	\end{align} 
	where $d$ is a metric.
\end{lemma}
\begin{proof}
	In the following, we prove the first lower bound, as the second follows analogously. Let us assume a simplified definition of the Wasserstein distance where we remove the constraint given by Eqn.~(\ref{constraint1}). This yields the optimization problem
	\begin{align}
	\label{equ:lower_bound}
	\tilde{d}_{\mathrm{W}}^{(p)} (X, Y) \define \min_{\tilde{C}} \left( \sum_{i=1}^{m} \sum_{j=1}^{n} \tilde{c}_{ij} \, d(x_i, y_j)^p \right)^{\frac{1}{p}}, 
	\end{align}
	where each $\tilde{C}$ is an $m \times n$ transportation matrix, so that its elements $\tilde{c}_{ij}$ satisfy
	\begin{align}
	\tilde{c}_{ij} &\geq 0 &&\text{for } 1 \leq i \leq m, 1 \leq j \leq n, \\
	\sum_{j=1}^{n} \tilde{c}_{ij} &= \frac{1}{m} &&\text{for } 1 \leq i \leq m.
	\end{align}
	Intuitively, this must be a lower bound, as all solutions to the original Wasserstein distance remain a feasible solution despite removing the constraint. 
	
	In the following, we show the solution for $\left( \tilde{d}_{\mathrm{W}}^{(p)}(X, Y) \right) ^p$ which is also the solution for $\tilde{d}_{\mathrm{W}}^{(p)}(X, Y)$. Let $\tilde{C}$ be any feasible matrix for Eqn.~(\ref{equ:lower_bound}) with elements $\tilde{c}_{ij}$, then
	\begin{align}
	& \sum_{i=1}^{m} \sum_{j=1}^{n} \tilde{c}_{ij} \, d(x_i, y_j)^p, \\
	&\geq \sum_{i=1}^{m} \sum_{j=1}^{n} \min_{k=1, \ldots, n} \tilde{c}_{ij} \, d(x_i, y_k)^p , \\
	&= \sum_{i=1}^{m} \min_{k=1, \ldots, n} d(x_i, y_k)^p \sum_{j=1}^{n} \tilde{c}_{ij}, \\
	&= \sum_{i=1}^{m} \min_{k=1, \ldots, n} d(x_i, y_k)^p \frac{1}{m}.
	\end{align}
\end{proof}

\subsubsection{Computation Scheme}

By leveraging the previous lemma, we come up with the following approach for approximating the Wasserstein distance. That is, we compute a lower bound $\mathit{LB}$ and, only if that is promising, compute the exact Wasserstein distance. 

A trivial approach would be based directly on the above lemma and obtain a corresponding lower bound by considering either the first or the second constraint, \ie, $\mathit{LB}_1$ and $\mathit{LB}_2$. For a vector $x_i$, these bound are obtained by computing the nearest neighbor via 
\begin{align}
\mathit{LB}_1 = \argmin_{k=1, \ldots, n} d(x_i, y_k)^p \quad\text{and} \mathit{LB}_2 = \argmin_{k=1, \ldots, m} d(x_k, y_j)^p , 
\end{align}
where the nearest neighbor search in the second is simply reversed. 

By taking the maximum of the two lower bounds, we obtain an even tighter lower bound 
\begin{equation}
\mathit{LB}^\ast = \max \{ \mathit{LB}_1, \mathit{LB}_2 \} .
\end{equation}
If the lower bound $\mathit{LB}^\ast$ between two baskets $X$ and $Y$ exceeds the distance of previous nearest neighbors, we skip computation of the exact Wasserstein distance. Overall, we significantly reduce the computation time in comparison to the original $O(l^3 \log(l))$ for the exact Wasserstein distance, while our nearest neighbor search is superior in the sense that it achieves a time complexity of $O(l^2)$. 

\subsection{Wasserstein-Based Dynamic Time Warping}

We have previously measured the distance between individual market baskets, while we now proceed by extending it to distances between \emph{sequences} of market baskets. Specifically, we now are interested in locating similar purchase histories. For this purpose, we utilize dynamic time warping and review its na{\"i}ve form in the following. Afterwards, we extend it to subsequences, while we finally combine it with a $k$-nearest neighbor procedure. 

Dynamic Time Warping has proven to be an exceptionally powerful distance measuring device for time series \cite{Kate.2016, Varatharajan.2018}. It is an algorithm for measuring similarity between two temporal sequences where events occur at different speeds. This is addressed by ``warping'' sequences non-linearly in the time dimension in order to determine a measure of similarity independent of certain non-linear variations in time. Mathematically, its distance between two sequences is the sum of the matched objects from the sequences after the two sequences have been optimally mapped.

The abovementioned arguments make the dynamic time warping framework a natural choice for measuring similarities between purchasing histories. Customers alike might make repeat purchases at different repeat rates. On top of that, might even exchange the original item against a different product which acts as a substitute. 

\subsubsection{Na{\"i}ve Dynamic Time Warping}

Consider two purchase histories of customer $c$ and $d$, namely, $B_c = [b_c^1, b_c^2, \ldots, b_c^n]$ and $B_d = [b_d^1, b_d^2, \ldots, b_d^m]$. According to dynamic time warping, the distance between them, $d_{\text{DTW}}(B_c, B_d)$, is determined by a matrix $D \in \mathbb{R}^{n \times m}$. For $i=1, \ldots, n$ and $j=1, \ldots, m$, the element $D_{ij}$ of $D$ is computed according to the recursive scheme
\begin{equation}
D_{ij} = d^{(p)}_{\mathrm{W}} (b_c^i, b_d^j) + \min \{ D_{i, j - 1}, D_{i - 1, j}, D_{i - 1, j - 1} \} 
\end{equation}
with
\begin{align}
D_{0,0} &= 0, \\
D_{i,0} &= D_{0,j} = \infty. \\
\end{align}
Evidently, na{\"i}ve dynamic time warping can compute similarities even when the two sequences are related non-linearly. However, it has a particular caveat as both the first and last element of the sequences must be identical matches. This is not beneficial in our case: \eg, a younger customer with short purchase history might map well with only some recent time span of a customer with a much longer history. Hence, we now extend the similarity matching to subsequential cases.

\subsubsection{Subsequential Dynamic Time Warping}

For a purchase history $B_c = [b_c^1, b_c^2, \ldots, b_c^n]$ and a potentially much longer purchase history $B_d = [b_d^1, b_d^2, \ldots, b_d^m]$, one could find a good matching between the sequence $B_c$ and a subsequence $B_d[i_s:i_e]$ which would be sufficient for our approach as we are later only interested in $B_d[i_e + 1]$. 

Fig.~\ref{fig:dtw} illustrates the idea of matching a subsequence of $B_d$. Here $B_c$ matches well with a subsequence of $B_d$ but would have large distance when start and end point of the sequences would be aligned. 

We are now left how subsequential matching can be accomplished. A trivial solution approach would be to consider all possible subsequences of $B_d$, \ie, $B_d[i_s:i_e]$ for $1 \leq i_s \leq i_e \leq m$. However, this approach would result in a time complexity of $O(nm^2)$ that renders applications with large-scale data infeasible. Instead, we adapt the DTW algorithm to identify matching subsequences by utilizing star-padding~\cite{Sakurai.2007}. The idea behind star-padding is to insert a special zero-distance value at the beginning of $B_c$, \ie, it has zero distance to $b_d^i$ for $i=1,2, \ldots, m$. This translates into the following subsequential DTW algorithm: the distance $d_{\text{SDTW}}(B_c, B_d)$ is determined by $\min_i D_{in}$ of matrix $D$, whose entries are computed recursively via
\begin{equation}
D_{ij} = d^{(p)}_{\mathrm{W}} (b_c^i, b_d^j) + \min \{ D_{i, j - 1}, D_{i - 1, j}, D_{i - 1, j - 1} \} 
\end{equation}
for all $i=1, \ldots, n$ and $j=1, \ldots, m$ with
\begin{align}
D_{i, 0} &= D_{0,0} = 0, \quad\text{and}\\
D_{0, j} &= \infty .
\end{align}
Here the use of star-padding~\cite{Sakurai.2007} reduces computational complexity to $O((n+1)m)$ for deriving the minimum subsequence distance.

\begin{figure}
	\centering
	\includegraphics[width=0.5\linewidth]{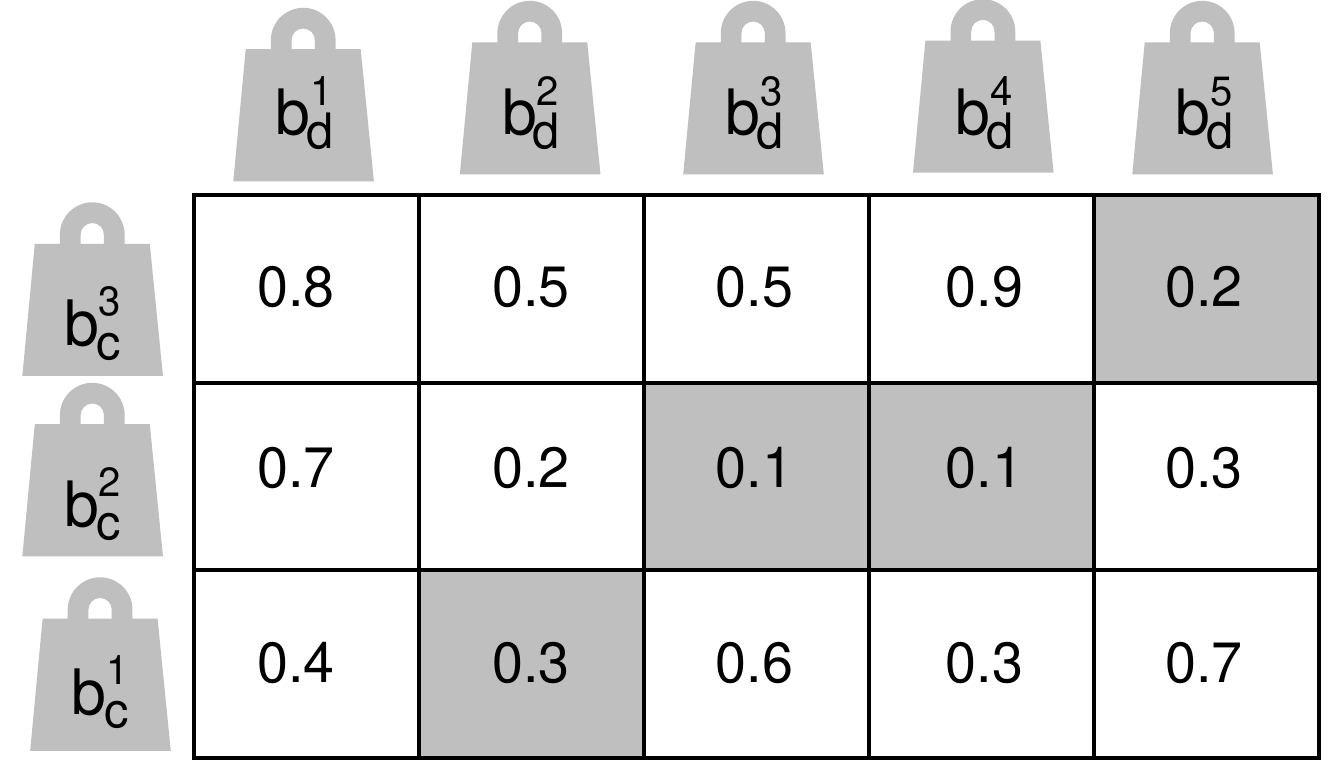}
	\caption{Example how subsequential dynamic time warping computes the distance between a given purchase history for customer $c$, \ie, $B_c = [b_c^{1}, b_c^{2}, b_c^{3}]$, and potential subsequences of the purchase history of customer $d$, $B_d = [b_d^{1}, b_d^{2}, b_d^{3}, b_d^{4}, b_d^{5}]$. The value in cell belonging to row $b_c^{i}$ and column $b_d^{j}$ denotes the Wasserstein distance between the two baskets in our example. The smallest distance is between $B_c$ and $B_d[2:5]$. Here, the market basket $b_c^2$ is fairly similar to both $b_d^3$ and $b_d^4$, possibly because customer $d$ has bought a similar basket twice.}
	\label{fig:dtw}
\end{figure}

\subsection{Prediction}

After having introduced a tool for computing distances between purchase histories, we now proceed by stating how we derive a prediction of the next market basket for a customer. 

\subsubsection{$k$-Nearest Neighbor Matching}

Instead of making predictions from the most similar subsequence, we now suggest a classifier that leverages information from multiple similar subsequences based on $k$-nearest neighbor. In general, the input to the $k$-NN approach consists of one unknown datapoint and the $k$ closest training examples in the feature space. After having identified the $k$ closest samples, it is widespread to infer the prediction by taking the majority class within the set of the $k$ neighbors.  

In order to utilize a nearest neighbor classifier for predicting market baskets, we adapt the method by making the following changes. First, we utilize the distance of the SDTW approach to define how close two samples are. Given a purchase history $B_c = [b_c^{1}, b_c^{2}, \ldots, b_c^{m}]$ as input, we then compute the prediction, \ie, the next market basket for a customer $c$, by considering all purchase histories $\mathcal{B} = \{ B_{1}, B_{2}, \ldots,  B_{k} \}$ for $k$ customers. We solve
\begin{align}
( i^*, i_s^*, i_e^* ) = \argmin_{\substack{1 \leq i_s \leq i_e \\ i=1, \ldots, k}} d_{\text{SDTW}}(B_c, B_{i}[i_s: i_e])
\end{align}
in order to find the closest (and potentially subsequential) purchase histories within the complete customer base. Note that this optimization problem is extended for the desired number of nearest neighbors.  

For $1$-nearest neighbor, the prediction for the next market basket is simply $B_{i^*}[i_e^* + 1]$. For $k > 1$, we take the most common items within the $k$ market baskets, where the size of the predicted basket is determined by the nearest neighbor $B_{i^*}[i_e^* + 1]$.

\subsubsection{Similarity-Aware Prediction}

Our Wasserstein-based approach has the advantage of returning a quantity of how similar the most similar purchase history is. We additionally leverage this in order to discard our predictions based on a look-up of similar purchase histories. Instead, we revert to a fallback prediction in the form of a straightforward heuristic that is often more reliable for such cases. This approach lets us exclude customers that have highly particular or unique habits and where such similarity scoring would otherwise fail. As a remedy, we then apply a simple heuristic for purchasing personal top-$n$ items; cf. baselines later for details. 

Our above switch is based on a simple condition $d^{(p)}_{\mathrm{W}}(B_c, B_d) < \tau$ where $B_d$ is a potential nearest neighbor and $\tau$ denotes a threshold. More precisely, we utilize our similarity-based prediction only as long as the $k$ nearest neighbors have an average distance of below $\tau$. Finally, the threshold $\tau$ presents a hyperparameter that is tuned as detailed in the online appendix. 

\section{Numerical Experiments}

\subsection{Performance Metrics}

For each predicted market basket, we evaluate the agreement of the predicted basket and the real basket using the following metrics:
\begin{itemize}[leftmargin=0.3cm]
\item \textbf{Wasserstein distance:} The Wasserstein distance is based on item embeddings and thus measures the degree to which the prediction is similar while accounting for products that act as substitutes. For instance, a market basket comprising ``red wine'' has a small Wasserstein distance to a market basket containing ``white wine''.
\item \textbf{F1-score:} The F1-score is defined as the harmonic mean of precision and recall and is a common metric for market basket comparison.
\item \textbf{Jaccard coefficient:} The Jaccard coefficient is the ratio of co-occurrences to non-co-occurrences between items of the predicted market basket and items of the true next market basket, \ie, between market basket $b^*$ and $b$. Formally, the Jaccard coefficient is defined by
\begin{equation}
	J = \frac{p}{p + q + r},
\end{equation}
where $p$ is number of items in both $b^*$ and $b$; $q$ is the number of items in $b^*$ and not in $b$; and $r$ the number of items not in $b^*$ but in $b$. 
\end{itemize}
Note that the F1-score and the Jaccard coefficient measure the overlap by counting the exact matches between market baskets, whereas the Wasserstein distance introduces a probabilistic distance.

\subsection{Datasets}

In this work, we experiment with three public datasets:
\begin{enumerate}[leftmargin=0.5cm]
	\item \textbf{Simplified Instacart:} We retrieved the Instacart groceries dataset for analyzing market baskets.\footnote{\url{https://www.kaggle.com/c/instacart-market-basket-analysis/data}} This dataset includes 3,214,874 orders with 49,688 different products grouped into 134 categories (\ie, aisles). We first start with this simpler task where, due to the category-based aggregation, the available assortment becomes considerably smaller and where we expect only small differences over our baselines. We included this analysis intentionally in our paper as it later allows us to discuss strengths and weaknesses of our method.  
	\item \textbf{Product-level Instacart:} We use the same dataset but afterwards study it at the product level. This greatly increases the number of available items and, as we shall see later, immediately renders our baselines inferior. Here we report our results for the items that are among the 500 most frequently purchased products. 
	\item \textbf{Ta-Feng grocery dataset:} The Ta-Feng dataset is a prevalent baseline dataset in marketing analytics that covers products from food, office supplies, and furniture.\footnote{\url{http://www.bigdatalab.ac.cn}} It contains 817,741 transactions that belong to 32,266 users and 23,812 items. Similarly to the previous dataset, we report our results for the items that are among the 500 most frequently purchased products.
\end{enumerate}
Our appendix lists additional summary statistics, as well as our preprocessing steps.

\subsection{Considered Baselines}

This section briefly summarizes our set of baseline models. Our choice was made in line with the status quo in prior literature \cite{Cumby.2004, Rendle.2010, Wang.2015, Yu.2016} as follows:
\begin{itemize}[leftmargin=0.3cm]
\item \textbf{Global top items:} As a simple baseline, we compute the performance of a model that predicts the $n_c$ most frequently purchased items across all customers. Note that $n_c$ depends on the customer.
\item \textbf{Personal top items:} This baseline purchases the top-$n_c$ items from customer $c$. Hence, it adapts to the specific taste of customer $c$, but cannot recognize time variations.
\item \textbf{Repurchase last basket:} The next baseline method predicts a market basket that consists of products bought in the customer's previous purchase, \ie, for a sequence of market baskets $B_c = [b_c^{1}, b_c^{2}, \ldots, b_c^{m}]$ the prediction is given by $b_c^{m + 1} \define b_c^{m}$. In contrast to the other models, this baseline predicts the size of the market as the size of the previous basket and is further capable to recognize simple autoregressive dynamics. 
\item \textbf{Association rules:} A common approach to analyze market baskets are association rules \cite{Agrawal.1994, Guidotti.2017, Guidotti.2018}. In general, association rules are used for investigating products that are often purchased together. As we are interested in a prediction of future products, we adapt the association rule in the following way. For a sequence $B_c = [b_c^{1}, b_c^{2}, \ldots, b_c^{m}]$ of market baskets, we build the tuple of product pairs from the Cartesian product for all consecutive baskets, \ie, 
\begin{equation}
	\mathcal{C}_c = \{(a,b) \text{ for } a,b \text{ in } b_c^{i} \times b_c^{i + 1} \text{ for } 1 \leq i \leq m - 1\}
\end{equation}
Then, we build the union of all sets of Cartesian products $\mathcal{C} = \bigcup_c \mathcal{C}_c$ and compute a matrix $\mathcal{S} \in \mathbb{N}^{n \times n}$, where $\mathcal{S}_{ab}$ denotes the number of times, item $a$ is in market basket $b_c^{i}$ and $b$ is in the subsequent market basket $b_c^{i + 1}$, \ie the support of tuple $(a,b)$ in $\mathcal{C}$. Here $n$ refers to the number of items in our assortment. The prediction of the next market basket of a sequence $B_c = [b_c^{1}, b_c^{2}, \ldots, b_c^{m}]$ is then given by the products that have been bought most often after products from basket $b_c^{m}$ have been bought.
\end{itemize}

\subsection{Prediction Performance}

\subsubsection{Simplified Category-Based Instacart Purchases}

Tbl.~\ref{tbl:aisle_results} lists predictive performance for our models when predicting market baskets on aisle level. Note that this is an artificial example where the complexity of the assortment was reduced in order to better understand the strengths and weaknesses of the different methods. 

Our approach outperforms the former baselines in terms of the F1-score and, for all other metrics, it is fairly on par. As a direct implication, we see that such a small dataset with more repeat purchasing (due to the fact that products are combined into categories) is not favorable for our approach, since it can apparently not leverage the hierarchical structure from the item embeddings.

\begin{table}[H]
	\centering
	\footnotesize
	\begin{tabular}{l ccccc}
		\toprule
		\textbf{Model} & \textbf{Wasserstein} & \textbf{F1-Score} & \textbf{Jaccard} \\
		& \textbf{distance} & & \textbf{coefficient} \\
		\midrule
		Global top items & 6.642 & 0.387 & 0.428 \\
		Personal top items & 5.799 & 0.493 & \bfseries 0.567 \\
		Repurchase last basket & \bfseries 5.699 & 0.483 & 0.465 \\
		Association rules & 7.108 & 0.269 & 0.297 \\
		\midrule
		Our approach & 5.704 & \bfseries 0.497 & \bfseries 0.567 \\
		\bottomrule
	\end{tabular}
	\caption{Performance for prediction of next market basket on a simplified variant of the Instacart dataset, where products were aggregated at aisle level. Accordingly, the complexity of the task is reduced in comparison to the other product-level datasets and, as expected, we observe a setting that is beneficial for the baselines. Best performance in bold.}
	\label{tbl:aisle_results}
\end{table}

\subsubsection{Product-Level Instacart Purchases}

We now proceed to evaluate the model for predicting market baskets on product level. The corresponding results are detailed in~Tbl.~\ref{tbl:product_results}. Among the baseline models, predicting the personal top items of each customer performs best, outperforming the second best baseline by 18.73\,\% in terms of F1-score. When comparing with our simplified dataset from the previous section, we evidently note based on the dataset that a good solution is considerably more complex as the repeat purchase heuristic fails to identify an accurate structure. In contrast, our approach outperforms all baselines: it increases the prediction performance in terms of F1-score by 2.54\,\%. 

\begin{table}[H]
	\centering
	\footnotesize
	\begin{tabular}{l ccc}
		\toprule
		\textbf{Model} & \textbf{Wasserstein} & \textbf{F1-Score} & \textbf{Jaccard} \\
		& \textbf{distance} & & \textbf{coefficient} \\
		\midrule
		Global top items & 11.483 & 0.174 & 0.186 \\
		Personal top items & 8.477 & 0.393 & 0.399 \\
		Repurchase last basket & 9.212 & 0.331 & 0.330 \\
		Association rules & 12.601 & 0.033 & 0.033 \\
		\midrule
		Our approach & \bfseries 8.348 & \bfseries 0.403 & \bfseries 0.407 \\
		\bottomrule
	\end{tabular}
	\caption{Performance for prediction of next market basket on Instacart dataset at product level. Best performance in bold.}
	\label{tbl:product_results}
\end{table}

\subsubsection{Ta-Feng Grocery Dataset}

Tbl.~\ref{tbl:ta_feng} lists the results for the Ta-Feng grocery dataset. For this dataset, all baseline models almost completely fail to find patterns. The strongest baseline model is again attained by choosing the personal top items, yielding an F1-score of 0.053 and a Jaccard coefficient of 0.259. In contrast, our approach benefits from learning more complex dynamics and thus bolsters the performance: it consistently outperforms all baseline models across all metrics. For instance, the F1-score amounts to 0.263, which corresponds to an improvement by almost 400\,\%. 

\begin{table}[H]
	\centering
	\footnotesize
	\begin{tabular}{l ccccc}
		\toprule
		\textbf{Model} & \textbf{Wasserstein} & \textbf{F1-Score} & \textbf{Jaccard} \\
		& \textbf{distance} & & \textbf{coefficient} \\
		\midrule
		Global top items & 33.910 & 0.047 & 0.152 \\
		Personal top items & 35.376 & 0.053 & 0.259 \\
		Repurchase last basket & 35.876 & 0.040 & 0.075 \\
		Association rules & 35.730 & 0.022 & 0.087 \\
		\midrule
		Our approach & \bfseries 29.661 & \bfseries 0.263 & \bfseries 0.550 \\
		\bottomrule
	\end{tabular}
	\caption{Performance for prediction of next market basket on Ta-Feng dataset. Best performance in bold.}
	\label{tbl:ta_feng}
\end{table}

\subsection{Sensitivity Analysis across Product Categories}

Tbl.~\ref{tbl:product_categories} lists prediction performance of our approach across different product categories. It enables us to judge whether product categories with strong customer loyalty and thus less likelihood of substitution effects can be better predicted. Evidently, the prediction of products from the categories ``bakery'' and ``pets'' achieves the highest prediction performance in terms of F1-score. In contrast, ``snacks'' and ``alcohol'' appear the hardest to predict. For instance, purchasing behavior of alcoholic beverages is known to be largely driven by price sensitivity and thus rare repeat purchasing with strong substitution effects. 

\begin{table}[H]
	\centering
	\footnotesize
	\begin{tabular}{l ccc}
		\toprule
		\textbf{Product} & \textbf{Wasserstein} & \textbf{F1-Score} & \textbf{Jaccard} \\
		\textbf{categories} & \textbf{distance} & & \textbf{coefficient} \\
		\midrule
		Bakery & 2.679 & 0.779 & 0.849 \\
		Pets & 1.200 & 0.725 & 0.732 \\
		Pantry & 3.853 & 0.615 & 0.607 \\
		Beverages & 4.804 & 0.576 & 0.583 \\ 
		Dairy eggs & 5.957 & 0.545 & 0.564 \\
		Deli & 5.202 & 0.544 & 0.583 \\
		Babies & 4.641 & 0.516 & 0.522 \\
		Frozen & 5.039 & 0.501 & 0.499 \\
		Household & 3.767 & 0.500 & 0.521 \\
		Breakfast & 4.306 & 0.498 & 0.520 \\
		Produce & 7.540 & 0.485 & 0.484 \\
		Snacks & 5.146 & 0.480 & 0.468 \\
		Alcohol & 3.544 & 0.423 & 0.472 \\
		\bottomrule
	\end{tabular}
	\caption{Sensitivity of the prediction performance by product category (sorted according to F1-score).}
	\label{tbl:product_categories}
\end{table}

\subsection{Computational Performance}

We are now interested in comparing the runtime of the exact Wasserstein distance against our approximation scheme. The following results are based on an Intel Xeon Silver 4114 2.2\,GHz processor with 10 cores and 64\,GB RAM. 

\textbf{Runtime:} We find that computation of the lower bound is roughly twice as fast as the exact scheme when evaluating it on the on product-level Instacart product dataset (as it is the largest one). Here the timings amount to 152 vs. 67.5 microseconds. 

\textbf{Hitrate of lower bound:} We also measured the hitrate of the lower bound which denotes the percentage of purchase histories that could be filtered out because the lower bound exceeded the Wasserstein distance to the closest candidates. For the previous dataset, we obtained a hitrate of 80.14\,\%. Hence, most candidates can be ignored without computing the exact Wasserstein distance. This result highlights that the lower bound is very tight to the actual Wasserstein distance.  

\textbf{Memory:} As $k$NN-SDTW merely utilizes dynamic time warping to compute distances between short sequences representing purchase histories, the memory requirement is neglectable.

\section{Use Cases for Personalized Purchase Predictions}

Currently, there are three main application areas for market basket predictions: 

\textbf{Recommender systems:} The majority of recommender systems currently used in real online shopping environments are based on collaborative filtering methods \cite{He.2017, Linden.2003, Sarwar.2001}. However, these methods usually produce only a list of items, out of which many are substitutes and rarely complementary. If complementary items are desired, our methods could help in offering not only a prediction of a \emph{single} item but make suggestions of item \emph{bundles} that are likely to be purchased together. For instance, a recommender system for cooking must suggest ingredients for a meal, for which the combination gives a nice taste. Hence, we provide a path to better identify complementary items for such settings. 

\textbf{Supply chain optimization:} Online shopping platforms and even brick-and-mortar supplies have recently started to experiment with preemptive delivery, \ie, preparing products for shipping even though the product has not yet been officially sold. However, these approaches are limited to single-item shipments and can not yet handle item bundles. Here our method provides a remedy: it makes suggestions of the actual market basket that is desired and should be shipped (as opposed to a single item thereof). 

\textbf{Assortment optimization:} When optimizing assortments, practitioners oftentimes focus on selecting products that are most profitable but ignore the importance of repeat purchases and particularly co-purchasing patterns. Here our approach can potentially provide new insights that help in obtaining a better understanding, as our approach allows online/offline store managers to simulate purchasing behavior while considering actual baskets. Hence, products that are bought together can be arranged appropriately, \ie, in close proximity. 
\section{Discussion}

\subsection{Relationship to Literature}

Personalized purchase predictions of market baskets are a challenging undertaking, which prohibits a na{\"i}ve application of traditional machine learning frameworks. Recurrent neural networks are effective at learning sequences, yet they return an output vector of a fixed size and can thus not adapt to the variable size of market baskets. Ranking is commonly utilized when predicting a precision@$k$ such as relevant when inferring subset of items, yet it lacks a rule that prescribes a cut-off point in order to determine the size of the market basket. Furthermore, it struggles with extreme size imbalances as in our case ($\left|b_c^{m}\right| \ll n$), merely orders items by purchase probability, and, therefore, does not learn substitution effects among items within a market basket (\eg, milk from brand $B$ might be purchased due to a promotion and thus replaces brand $A$ in the market basket).

\subsection{Potential for Future Work}

As other methods, ours is not free of limitations. In order to better understand the strengths and weaknesses of our method, we performed an additional numerical comparison on a simplified Instacart dataset where products were artificially aggregated across categories (\ie, aisles). Evidently, the small size of the assortment, the dominance of repeat purchases, and the lack of strong dynamics rule out potential benefits from our approach and thus let it appear on par with baselines. A potential explanation stems from the fact that the use of embeddings is almost redundant and that the aggregation cancels out any complex dynamics that could be captured by sequence matching. On the contrary, our proposed method is particularly strong on complex tasks with large assortments, showing that its performance is considerably superior. We also see that it is difficult for our method to capture complex substitution effects (especially, when these are driven by spontaneous purchases or price promotions as common for alcoholic beverages).  

\subsection{Concluding Remarks}

This paper advances the existing literature on sales forecasting which is primarily studying purchase predictions at item or firm level, whereas we propose an innovative approach at basket level for personalized purchase predictions. Here the objective is to identify the \emph{exact} set of items (\ie, the market basket) from a future purchase, while drawing upon past purchase histories as input. Our prediction algorithm entails several unique characteristics: (1)~It is based on dense (embedded) vector representations of items in order to infer hidden structures among them, which eventually aids the identification of co-occurence patterns. (2)~We propose to utilize the Wasserstein distance for this task based on which we measure the similarity between baskets. It implies intuitive interpretations and further quantifies predictability. (3)~We perform cross-customer similarity matching. Thereby, we deliberately search for similar purchasing habits in the complete set of shopping trajectories. This thus facilitates a richer knowledge base and differs from the recurrent or autoregressive models in the literature which merely extrapolate the historic time series of an individual customer without information from other customers. (4)~We develop a fast prediction algorithm based on $k$-nearest neighbor subsequential dynamic time warping, derive its computational complexity, and suggest a lower bound for the Wasserstein distance. 

\bibliographystyle{ACM-Reference-Format}
\bibliography{literature}

\appendix

\section{Research Methods}

\subsection{Summary statistics}
Tbl~\ref{tbl:datasets} lists the three datasets utilized for our experiments along with statistics about the size and complexity of the dataset. 

\begin{table}[H]
	\centering
	\footnotesize
	\begin{tabular}{l ccc}
		\toprule
		\textbf{Dataset} & \textbf{Customers} & \textbf{Baskets} & \textbf{Products} \\
		\midrule
		Simplified Instacard & 65710 & 1634548 & 134 \\	
		Product-level Instacart & 27139 & 603457 & 500 \\
		Ta-Feng grocery dataset & 9451 & 172086 & 500 \\
		\bottomrule
	\end{tabular}
	\caption{Summary statistics of the datasets (after preprocessing).}
	\label{tbl:datasets}
\end{table}

\subsection{Preprocessing}

We only consider customers that have a purchase history of at least ten market baskets, each with at least five items. Further, we split our datasets into training, validation and test set in ratios of 80\,\%, 10\,\%, and 10\,\%, respectively. In line with prior literature, we adopt a customer-level splitting strategy. That is, we split by customers, \ie, each customer is in one of the different sets but not in two. For each customer $c$ in the test set, we utilize the purchase history $B_c = \left[b_1^c, \ldots, b_m^c \right]$ for identifying the customers with the closest purchase histories in our training set. By taking their next purchase as the gold standard, we then predict the next market basket of customer $c$ in order to measure the performance.  

\subsection{Estimation Details}

Our approach contains only very few parameters, which contribute to the robustness of our evaluation. All our hyperparameters are listed in Tbl.~\ref{tbl:parameter_choice}. In fact, our approach only involves three hyperparameters: (i)~the embedding dimensions for products, (ii)~the threshold $\tau$ that decides if to utilize a fallback prediction or our approach, and (iii)~the number of $k$-nearest neighbors that we base our prediction on. The embedding dimension was set to 50 across all experiments as we found this already beneficial. 

State-of-the-art methods~\cite{Cumby.2004, Rendle.2010, Wang.2015, Yu.2016} must fix the size of the predicted basket. Here choices of $n=5$ or $n=10$ are common; however, we found that this impedes the performance given the large assortment in our experiments. Hence, we decided to introduce additional flexibility by making the size of the predicted market basket dynamic to the customer's personal behavior. Hence, all baselines, as well as the fallback heuristic in our proposed model, utilize a customer-specific size $n_{c}$ that amounts to the average basket size of the previous purchases from customer $c$.

\begin{table}[H]
	\centering
	\footnotesize
	\begin{tabular}{l r}
		\toprule
		\textbf{Parameter} & \textbf{Tuning range} \\
		\midrule
		Product embedding dimensions & 50 (fixed) \\
		Threshold $\tau$  & [5,10, \ldots, 30,35] \\
		Number of $k$-nearest neighbors & [1,2,5,10,20]  \\
		\bottomrule
	\end{tabular}
	\caption{Hyperparameters with tuning ranges. Note that the small number of hyperparameters contributes to the robustness of our approach, as well as our evaluation.}
	\label{tbl:parameter_choice}
\end{table}

Our earlier descriptions were based on a Wasserstein distance with an arbitrary $p \geq 1$. For our experiments, we now have to specify a choice: we utilize the $\ell_1-$Wasserstein distance ($p = 1$) for measuring similarities between market baskets.

\end{document}